\documentclass[11pt]{article}
\usepackage{fullpage}
\usepackage{verbatim} 
\usepackage[usenames, dvipsnames]{xcolor}
\usepackage{graphicx}
\usepackage{subcaption}
\usepackage{bbm}
\usepackage{mathtools}
\usepackage{amsmath}
\usepackage{amsthm, thmtools}
\usepackage{amstext,amssymb, amsfonts}
\usepackage{empheq} 

\usepackage{lmodern} 

\usepackage{algorithmic,algorithm} 
\usepackage{todonotes}
\usepackage[letterpaper=true,colorlinks=true,pdfpagemode=none,urlcolor=Blue,linkcolor=Fuchsia,citecolor=OliveGreen,pdfstartview=FitH]{hyperref}
\usepackage{natbib}
\newcounter{mynotes}
\def\notes{1}
\newcommand{\gnote}[1]{\ifnum\notes=1{{\sf\color{blue} [Gopi: #1]}}\fi}

\declaretheorem[within=section]{theorem}
\declaretheorem[sibling=theorem]{corollary}

\declaretheorem[sibling=theorem]{lemma}
\declaretheorem[sibling=theorem]{claim}

\declaretheorem[sibling=theorem]{definition}

\declaretheorem[sibling=theorem]{proposition}
\declaretheorem[sibling=theorem]{remark}

\def\ba{{\mathbf a}}
\def\bb{{\mathbf b}}
\def\bc{{\mathbf c}}
\def\bd{{\mathbf d}}
\def\be{{\mathbf e}}

\def\br{{\mathbf r}}

\def\bw{{\mathbf w}}

\newcommand{\R}{\mathbb{R}} 

\newcommand{\cQ}{\mathcal Q}

\newcommand{\argmin}{\mathrm{argmin}} 

\newcommand{\inpro}[2]{\left\langle #1,#2 \right\rangle} 

\renewcommand{\epsilon}{\varepsilon}

  \newcommand{\beq}{\begin{equation}}
  \newcommand{\eeq}{\end{equation}}
  \newcommand{\beqn}{\begin{equation*}}
  \newcommand{\eeqn}{\end{equation*}}
  \newcommand{\beqr}{\begin{eqnarray}}
  \newcommand{\eeqr}{\end{eqnarray}}
  \newcommand{\beqrn}{\begin{eqnarray*}}
  \newcommand{\eeqrn}{\end{eqnarray*}}
  \newcommand{\bmline}{\begin{multline}}
  \newcommand{\emline}{\end{multline}}
  \newcommand{\bmlinen}{\begin{multline*}}
  \newcommand{\emlinen}{\end{multline*}}

\def \OPT {\mathsf{OPT}}
\def \dash {\rightarrow\hspace{1cm}}

\newcommand{\TOP}{\operatorname{\mathsf{TOP}}}
\newcommand{\multiTOP}{\operatorname{\mathsf{Multi-TOP}}}

\newcommand{\rank}{\operatorname{\mathsf{Rank}}}
\newcommand{\multirank}{\operatorname{\mathsf{MultiRank}}}

\newcommand{\cs}{\operatorname{\mathsf{cs}}}
\newcommand{\gcs}{\operatorname{\mathsf{gcs}}}
\newcommand{\co}{\operatorname{\mathsf{co}}}
\newcommand{\argmax}{\operatorname{argmax}}
\newcommand{\aij}{a_{ij}}
\newcommand{\bij}{b_{ij}}
\newcommand{\piij}{\pi_{ij}}

\newcommand{\grad}{\nabla}
\newcommand{\Qplus}{\cQ^{++}}
\newcommand{\convexhull}{\mathsf{ConvHull}}

\newcommand{\boldw}{\mathbf{w}}

\newcommand{\barw}{\mathbf{w}}
\newcommand{\bara}{\mathbf{a}}
\newcommand{\barb}{\mathbf{b}}
\newcommand{\barf}{f}
\newcommand{\barpi}{\pi}
\newcommand{\barS}{S}

\title{Ranking with Multiple Objectives}
\author{Nikhil R. Devanur\thanks{Amazon. Email: \texttt{Iam@nikhildevanur.com}. Work done while the author was at Microsoft Research.} \and Sivakanth Gopi\thanks{Microsoft Research. Email: \texttt{sigopi@microsoft.com}}}
\date{}
\begin{document}

\maketitle

\abstract{
In search and advertisement ranking, it is often required to simultaneously maximize multiple objectives. 
For example, the objectives can correspond to multiple intents of a search query, 
or in the context of advertising, they can be relevance and revenue. It is important to efficiently find rankings which strike a good balance between such objectives. 
Motivated by such applications, we formulate a general class of problems where
\begin{itemize}
	\item each result gets a different score corresponding to each objective, 
	\item the results of a ranking are aggregated by taking, for each objective, a weighted sum of the scores in the order of the ranking, and 
	\item an arbitrary concave function of the aggregates is maximized. 
\end{itemize}
Combining the aggregates using a concave function will naturally lead to more balanced outcomes.
We give an approximation algorithm in a bicriteria/resource augmentation setting: 
the algorithm with a slight advantage does as well as the optimum. 
In particular, if the aggregation step is just the sum of the top $k$ results,  
then the algorithm outputs $k+1$ results which do as well the as the optimal top $k$ results.
We show how this approach helps with balancing different objectives via simulations 
on synthetic data as well as on real data from LinkedIn. 
}

\thispagestyle{empty}
\addtocounter{page}{-1}

\newpage
\section{Introduction} 
We study the problem of ranking with multiple objectives. 
Ranking is an important component of many online platforms such as 
Google, Bing, Facebook, LinkedIn, 
Amazon, Yelp, and so on.
It is quite common that the platform has multiple objectives while choosing a ranking. 
For instance, in a search engine, when someone searches for ``jaguar'', 
it could refer to either the animal or the car company. 
Thus there is one set of results that are relevant for jaguar the animal, 
another for jaguar the car company, and the search engine has to 
pick a ranking to satisfy both intents. 

Another common reason to have multiple objectives is advertising. 
The final ranking produced has \emph{organic} results as well as \emph{ads}, 
and the objectives are relevance and revenue.
Ads contribute to both relevance and revenue, where as organic results 
only contribute to relevance.
While in some cases ads occupy specialized slots,  
it is becoming more common to have floating ads. 
Also, in many cases, the same result can qualify both as an ad 
and as an organic result, and it is not desirable to repeat it. 
In such cases, one has to produce a single ranking of all the results 
(the union of organic results and ads)
that achieves a certain tradeoff between the two objectives. 

The predominant methodology currently used to handle multiple objectives is 
to combine them into one objective using a linear combination \citep{vogt1999fusion}. 
The advantage of this is that it can trace out the entire pareto frontier of the achievable objectives. 
The disadvantage is that you have to choose one linear combination for 
a large number of instances. 
This often results in cases where one objective is favored much more than the others. This is illustrated in Figure \ref{fig:scatter}. To explain this figure we introduce some notation.

Suppose that there are $m$ \emph{instances}, and for each instance there are $n$ \emph{results} 
that are to be ranked. 
Each result $j$ for instance $i$ has two numbers associated with it, 
$\aij$ and $\bij$, that correspond to the two objectives.
Given a ranking $\pi: [n] \rightarrow [n]$, 
we aggregate the two objective values for instance $i$  
using \textit{cumulative scores} defined as 
$$\cs_\bw(\ba_i,\pi)=\sum_{j\in [n]} w_ja_{i\pi(j)}\text{ and }\cs_\bw(\bb_i,\pi)=\sum_{j\in [n]} w_jb_{i\pi(j)},$$ 
for some non-negative weight vector $\bw=(w_1,\cdots,w_n)$  
with (weakly) decreasing weights i.e. $w_1\ge w_2 \ge \dots \ge w_n \ge 0.$\footnote{Throughout the paper, we use the notation policy that $w_i$ is the $i^{\rm th}$ coordinate of $\boldw$. Further, $\ba_i$ is the vector $(a_{i1}, a_{i2}, \cdots, a_{in})$.} 
For example, in advertisement ranking, $w_i$ can represent the click rate i.e. the probability that a user clicks the $i^{\rm th}$ result in the ranking. It is natural that the click rates decrease with the position i.e. it is more probable that a top result is clicked. Suppose $a_j$ represents the revenue generated when $j^{\rm th}$ ad is clicked, and $b_j$ represents the relevance of the $j^{\rm th}$ ad to the user query. Then $\cs_\bw(\ba,\pi)$ represents the expected revenue generated and $\cs_\bw(\bb,\pi)$ represents the expected total relevance for the user when the ads are ranked according to $\pi.$ 
In the figure,  the weight vector is the one used in 
discounted cumulative gain (DCG), and normalized DCG (NDCG) \citep{burges2005learning}, 
which are  standard measures  often used in evaluating search engine rankings. 
This weight vector is: 
\begin{equation}
\label{eqn:weight}
w_i = \frac{1}{\log_2(i+1)} \enspace. 
\end{equation}
We normalize the cumulative scores by the best possible ranking for each objective. 
This is motivated by two things: the resulting numbers are all in $[0,1]$ so they are comparable to each other, 
and how well the ranking did relative to the best achievable one is often the more meaningful measure. 
 We define 
 $$\cs^*_\bw(\ba)=\max_\pi \cs_\bw(\ba,\pi) \text{ and } \cs^*_\bw(\bb)=\max_\pi \cs_\bw(\bb,\pi).$$ 
When the weight vector is the one mentioned above, 
we refer to the normalized cumulative scores as NDCG.
Figure \ref{fig:scatter} shows scatter plots of the NDCGs for the two objectives, 
for different algorithms: a given algorithm produces ranking $\pi_i$ for instance $i$,  and
each dot in the plot is a point 
$$\left(\frac{\cs_\bw(\ba_i,\pi_i)}{\cs^*_\bw(\ba_i)} , \frac{\cs_\bw(\bb_i,\pi_i)}{\cs^*_\bw(\bb_i)}\right).$$

The source of the data is LinkedIn news feed: it is a random sample from one day of results. 
Here we do not go into the details of what the two objectives are, etc.; 
Section \ref{sec:experiments} has more details. 
On the right, we show the result of ranking using the sum of the two scores. 
The triangle shape of the scatter plot is persistent across different samples, 
and different choices of linear combinations. 
What we wish to avoid are the two corners of the triangle
where one of the two NDCGs is rather small. 
Ideally, we like to be at the apex of the triangle which is at the top right corner of the figure. 

On the left, we show the results of our algorithm for the following objective: 
\[ 
\max_\pi \cs_\bw(\ba_i,\pi)\cdot \cs_\bw(\bb_i,\pi) .
\]
The undesirable corners of the triangle have vanished: 
instances where one objective is much smaller than the other are rare if any. 
The points are  closer to the top right corner.  
\begin{figure}
	\centering
	
	\includegraphics[width=.7\textwidth]{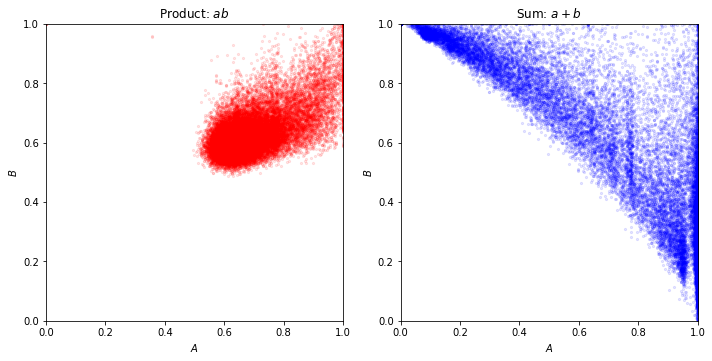}
	\caption{Scatter plot of NDCGs for two different objectives, $A$ and $B$, on real data.}
	\label{fig:scatter} 
\end{figure}

\subsection{Main Results}
The key idea is to combine the two cumulative scores using a \emph{concave} function $f$. (Maximizing the product is the same as maximizing the sum of logs, which is a concave function.)
Concave functions tend to favor more balanced outcomes almost by definition: the function at the average of two points is 
at least as high as the average of the function at the two points, i.e., $f(\tfrac{x_1 + x_2}{2},\tfrac{y_1 + y_2}{2} ) \ge \tfrac{f(x_1,y_1)  + f(x_2,y_2)}{2}$. 
Figure \ref{fig:scatter} is a good demonstration of this. 

We allow arbitrary concave functions that are strictly increasing in each coordinate. 
We define the \textit{combined objective score} of $\pi$ with weights $\bw$ as $$\co_{\bw,f}(\ba,\bb,\pi)=f\left(\cs_\bw(\ba,\pi),\cs_\bw(\bb,\pi)\right).$$
For some objectives, the sum of the cumulative scores 
across different instances is still an important metric, 
e.g., the total revenue, or the number of clicks, etc.  
We allow incorporating such metrics via a global concave function, 
i.e., a concave function of the sum of all the cumulative scores over all the instances. 
Let $F$, and $f_i$ for $i \in [n]$ be concave functions in 2 variables. 
We consider the problem of finding a ranking $\pi_i$ for each $i$ in order to maximize 
\[ 
\sum_{i\in [n]} \co_{\bw,f_i}(\ba_i,\bb_i,\pi_i) + 
F\left( \sum_{i\in [n]}\cs_{\bw}(\ba_i,\pi_i), \sum_{i\in [n]}\cs_{\bw}(\bb_i,\pi_i)  \right). 
\]  

Our main results are polynomial time bi-criteria approximation algorithms for the problem mentioned above. 
{These are similar in spirit to results with resource augmentation in scheduling, or Bulow-Klemperer style results in mechanism design.}  
\begin{itemize}
	\item Consider the special case where we sum the top $k < n $ entries, i.e., the weight vector is 
	$k$ ones followed by all zeros. 
	For this case, we allow the algorithm to sum the top $k+1$ entries, 
	and show that the resulting objective  
	 is at least as good as the optimum for the sum of the top $k$ results. 
	\item For the general case, the algorithm gets an advantage as follows: 
	replace one coordinate of $\bw$, say $w_j$, with the immediately preceding coordinate $w_{j-1}$. For the case of summing the top $k$ entries, this corresponds to replacing the $(k+1)^{\rm st}$ coordinate 
	which is a 0, with the $k^{\rm th}$ coordinate which is a 1.  
	The replacement can be different for different $i$. 
	Again, the ranking output by the algorithm with the new weights does as well as the optimum with the original weights $\bw$. See Theorem~\ref{thm:main_local} and Theorem~\ref{thm:main_global} for formal statements.
	
	One advantage of such a guarantee is that it does \emph{not} depend on the parameters of the convex functions, such as the Lipshitz constants, or the range of values, as is usual with other types of guarantees. This allows greater flexibility in the choice of these convex functions. 
	\item When there is no global function $F$, 
	for each $i$, the algorithm just does a binary search (Proposition~\ref{prop:binary_search_OPT}). 
	In each iteration of the binary search, we compute a ranking optimal for a linear combination of the two objectives. The running time to solve each ranking problem (i.e. each instance $i$) is $O(n\log^2 n)$. 
	In practice, the ranking algorithms are required to be very fast, so this is an important property. 
	For the general case, this is still true, provided that we are given 2 additional parameters that are optimized appropriately. In practice, such parameters are tuned `offline' so we can still use the binary search to rank `online' any new instance $i$. 
\end{itemize}


\subsection{Related Work}

\emph{Rank aggregation} is much studied, most frequently 
in the context of databases with `fuzzy' queries \citep{fagin1997incorporating} and in the context of information retrieval or web search \citep{dwork2001rank,aslam2001models}. 
There are two main categories of results \citep{renda2003web}, 
(i) where the input is a set of rankings \citep{dwork2001rank,renda2003web}, and 
(ii) where the input is a set of scores \citep{vogt1999fusion,fox1994combination}. 
Clearly score based aggregation methods are more powerful, 
since there is strictly more information; 
our paper falls in the score based aggregation category.
 
Among the score based methods,  
\citet{vogt1999fusion} uses the same form of cumulative scores as us, 
and empirically evaluate the usage of a \emph{linear} combination of cumulative scores. 
They identify limitations of this method and conditions under which this does well. 
\citet{fox1994combination} propose and evaluate several methods for combining the scores for different objectives result by result which are then used to rank. 
In contrast, we first aggregate the scores for each objective and 
then combine these cumulative scores. 

\citet{azar2009multiple} also consider rank aggregation motivated by multiple intents in search rankings, but with several differences. 
They consider a large number of different \emph{intents}, as opposed
to this paper where we focus on just 2.  
Their objective function also depends on a weight vector but in a different way. 
For each intent, a result is either `relevant' or not, 
and given a ranking, the cumulative score for that intent is 
the weight corresponding to the highest rank at which a relevant result appears. Their objective is a weighted sum of the cumulative scores across all intents. 

The rank based aggregation methods are closely related to voting schemes and social choice theory, and a lot of this has focused on algorithms to compute the Kemeny-Young rank aggregation \citep{young1978consistent, young1988condorcet,saari1995basic,borda1784memoire}.

\paragraph{Organization:}
In Section \ref{sec:local} we give our binary search based  algorithm for a single instance. 
The general case is presented in Section \ref{sec:global}. 
We present experimental results in Section \ref{sec:experiments}. Appendix contains some missing proofs.

\section{A single instance of ranking with multiple objectives}
\label{sec:local}


Let us formally define the $\rank$ problem.

\begin{definition}[$\rank(\ba,\bb,\bw,f)$]
Given $\ba,\bb,\bw,f$, find a ranking $\pi$ which maximizes the combined objective i.e. find $\co^*_\bw(\ba,\bb)=\max_\pi \co_{\bw}(\ba,\bb,\pi).$
\end{definition}

It is not clear if $\rank(\ba,\bb,\bw,f)$ can be efficiently solved, because it involves optimization over all rankings and there are exponentially many of them. Our main result shows that when $f$ is concave, we can find nearly optimal solutions. We will assume that $f$ is differentiable and has continuous derivatives.
\begin{theorem}
\label{thm:main_local}
Suppose $f(\alpha,\beta)$ is a concave function over the range $\alpha,\beta\ge 0$ and $f$ is strictly increasing in each coordinate in that range. Given an instance of $\rank(\ba,\bb,\bw,f)$, there is an algorithm that runs in $O(n\log^2 n)$ time\footnote{This is a Las Vegas algorithm i.e. always output the correct answer but runs in $O(n\log^2 n)$ time with high probability. We also give $O(n\log n\log B)$ algorithm when $a_i,b_i$ are integers bounded by $B$.} and outputs a ranking $\pi$ of $[n]$ such that $\co_{\bw'}(\ba,\bb,\pi) \ge \co^*_\bw(\ba,\bb)$ where $\bw'=\bw+(w_i-w_{i+1})\be_{i+1}$ for some $i\in [n-1]$. In other words, $\bw'$ is obtained by replacing $w_{i+1}$ with $w_i$ in $\bw$ for some $i\in [n-1]$. 
\end{theorem}

We have the following corollary for the important special case where the cumulative scores are the sum of scores of top $k$ elements, i.e., $\bw=(1,\dots,1,0,\dots,0)$ with exactly $k$ ones. In this special case, $\rank$ is called the $\TOP_k$ problem.
\begin{corollary}
	\label{cor:topk_local}
	Given such an instance of $\TOP_k(\ba,\bb,f)$, there is an efficient algorithm that outputs a subset $S\subset [n]$ of at most $k+1$ elements such that $$f\left(\sum_{i\in S}a_i,\sum_{i\in S}b_i\right) \ge \max_{|T|=k} f\left(\sum_{i\in T}a_i,\sum_{i\in T}b_i\right).$$
\end{corollary}

We will now prove Theorem~\ref{thm:main_global}. We also make the mild assumption that the numbers in $\ba,\bb,\bw$ are generic for the proof, which can be achieved by perturbing all the numbers with a tiny additive noise. In particular we will assume that $w_1>w_2>\dots>w_n>0$. This only perturbs $\co^*_\bw(\ba,\bb)$ by a tiny amount. By a limiting argument, this shouldn't affect the result.
To prove Theorem~\ref{thm:main_local}, we create a convex programming relaxation for $\co^*_\bw(\ba,\bb)$ as shown in~(\ref{eqn:local_primal}) and denote its value by $\OPT$.

\begin{equation}
\label{eqn:local_primal}
\begin{aligned}
\OPT=\max_{\pi_{ij},\alpha,\beta\ge 0}\ & f(\alpha,\beta)&\\
s.t.\ &  \alpha \le \sum_{i,j=1}^n \pi_{ij}w_ia_j   &\dash (p)\\
& \beta  \le \sum_{i,j=1}^n \pi_{ij}w_ib_j  &\dash (q)\\
&\forall i\ \sum_{j=1}^n \pi_{ij} \le 1  &\dash (r_i)\\
&\forall j\ \sum_{i=1}^n \pi_{ij} \le 1  &\dash (c_j)
\end{aligned}
\end{equation}
It is clear that $\OPT$ is a relaxation for $\co^*_\bw(\ba,\bb)$ with $\co^*_\bw(\ba,\bb)\le \OPT.$ By convex programming duality, $\OPT$ can be expressed as a dual minimization problem~(\ref{eqn:local_dual}) by introducing a dual variable for every constraint in the primal as shown in (\ref{eqn:local_primal}). Note that by Slater's condition, strong duality holds here~\cite{BoydV04}. The constraints in the dual correspond to variables in the primal as shown in~(\ref{eqn:local_dual}).
\begin{equation}
\label{eqn:local_dual}
\begin{aligned}
\OPT=\min_{r_i,c_j,p,q\ge 0}\ & \sum_{i=1}^n r_i+\sum_{j=1}^n c_j+f^*(-p,-q)&\\
s.t.\ &\forall i,j\ \ r_i+c_j\ge w_i(pa_j+qb_j) &\dash (\pi_{ij})\\
\end{aligned}
\end{equation}
Here $f^*$ is the Fenchel dual of $f$ defined as $$f^*(\mu,\nu)=\sup_{\alpha,\beta\ge 0} \left(\mu\alpha+\nu \beta +f(\alpha,\beta)\right).$$ Note that $f^*$ is a convex function since it is the supremum of linear functions. 
Since $f(\alpha,\beta)$ is strictly increasing in each coordinate, $f^*(\mu,\nu)=\infty$ unless $\mu,\nu< 0$. Since the dual is a minimization problem, the optimum value is attained only when $\mu,\nu < 0$. Hereafter, wlog, we  assume that $p,q> 0$ in the dual~(\ref{eqn:local_dual}). 
For example when $f(\alpha,\beta)=\log(\alpha\beta)$, 
\[
f^*(\mu,\nu)=
\begin{cases}
-\log(\mu\nu)-2 &\text{if }  \mu,\nu<0\\
\infty &\text{else}.
\end{cases}
\]

If $(\pi^*,\alpha^*,\beta^*)$ is some optimal solution for the primal~(\ref{eqn:local_primal}) and $(\br^*,\bc^*,p^*,q^*)$ is some optimal solution for the dual~(\ref{eqn:local_dual}), then they should together satisfy the KKT conditions given in~(\ref{eqn:KKT_conditions}). A constraint of primal is tight if the corresponding variable in the dual is strictly positive and vice-versa.
\begin{equation}
\label{eqn:KKT_conditions}
\begin{array}{ccl|ccl}
p^*>0 &\Rightarrow& \sum_{ij} \piij^* w_ia_j = \alpha^* & \nabla f(\alpha^*,\beta^*)&=&(p^*,q^*)\\
q^*>0 &\Rightarrow& \sum_{ij} \piij^* w_ib_j =\beta^* & \piij^*>0 &\Rightarrow& r_i^*+c_j^* = w_i(p^*a_j+q^*b_j)\\
r_i^*>0 &\Rightarrow& \sum_j \piij^* =1 & & & \\
c_j^*>0 &\Rightarrow& \sum_i \piij^* =1 & & &
\end{array}
\end{equation}

\begin{proposition}
\label{prop:dual_fixedpq}
Let $p,q> 0$ be fixed. Then the value of the minimization program in~(\ref{eqn:local_dual}) is given by $$\Psi(p,q)=\cs^*_\bw(p\ba+q\bb)+f^*(-p,-q)$$ where $p\ba+q\bb=(pa_1+qb_1,\dots,pa_n+qb_n).$ Moreover the KKT condition~(\ref{eqn:KKT_conditions}) can be simplified to:
\begin{equation}
	\label{eqn:simplified_KKT}
	\begin{aligned}
	&\pi^*\in \convexhull\{\pi: \cs^*_\bw(p^*\ba+q^*\bb)=\cs_\bw(p^*\ba+q^*\bb,\pi)\},\\
	&\alpha^*=\sum_{ij} \piij^*w_ia_j,\ \ \beta^*=\sum_{ij} \piij^*w_ib_j,\\
	&\grad f(\alpha^*,\beta^*)=(p^*,q^*).
\end{aligned}
\end{equation}
\end{proposition}
\begin{proof}
For a fixed $p,q> 0$, the dual program (\ref{eqn:local_dual}) reduces (after ignoring the fixed additive term $f^*(-p,-q)$) to the following linear program:
\begin{align*}
\min_{r_i,c_j\ge 0}\ & \sum_{i=1}^n r_i+\sum_{j=1}^n c_j&\\
s.t.\ &\forall i,j\ r_i+c_j\ge w_i(pa_j+qb_j)&\dash (\piij).
\end{align*}
The dual linear program is:
\begin{align*}
\max_{\pi_{ij}\ge 0}\ & \sum_{ij} \pi_{ij} w_i(pa_j+qb_j)\\
&\forall i\ \sum_{j=1}^n \pi_{ij} \le 1  &\dash (r_i)\\
&\forall j\ \sum_{i=1}^n \pi_{ij} \le 1  &\dash (c_j)
\end{align*}

The constraints on $\pi$ correspond to doubly stochastic constraints on the matrix $\pi$. Therefore by the Birkhoff-von Neumann theorem, the feasible solutions are convex combinations of permutations and the optimum is attained at a permutation. An optimal permutation should sort the values of $p\ba+q\bb$ in decreasing order and convex combinations of such permutations are also optimal solutions. Thus the set of solutions $\pi^*\in \convexhull\{\pi: \cs^*_\bw(p^*\ba+q^*\bb)=\cs_\bw(p^*\ba+q^*\bb,\pi)\}$ and the value of both the above programs is $\cs^*_\bw(p\ba+q\bb).$ 


\end{proof}

\begin{lemma}
\label{lem:atmost_twoties}
Fix some $p,q>0$. Then one of the following is true:
\begin{enumerate}
	\item There are no ties among $p\ba+q\bb$ i.e. there is a unique $\pi$ such that $\cs^*_\bw(\ba,\bb)=\cs_\bw(\ba,\bb,\pi)$
	\item There is exactly one tie among $p\ba+q\bb$ i.e. there are exactly two permutations $\pi_1,\pi_2$ such that $\cs^*_\bw(\ba,\bb)=\cs_\bw(\ba,\bb,\pi_1)=\cs_\bw(\ba,\bb,\pi_2)$. Moreover $\pi_1,\pi_2$ differ by an adjacent transposition i.e. $\pi_2$ can be obtained from $\pi_1$ by swapping adjacent elements.
\end{enumerate}
\end{lemma}

\begin{proof}
$\cs^*(p\ba+q\bb)=p\cs(\ba,\pi)+q\cs(\ba,\pi)$ where $\pi$ is any permutation which sorts $p\ba+q\bb$ in descending order. There are two cases:
	\begin{enumerate}
		\item[Case 1:]
	 If there are no ties among $p\ba+q\bb$, then the permutation $\pi$ is unique.
	 \item[Case 2:] Suppose there are ties among $p\ba+q\bb$. Because we assumed that $\ba,\bb$ are generic, there can be at most one tie among $(pa_1+qb_1,\dots,pa_n+qb_n)$ i.e. there is at most one pair $s,t$ such that $pa_s+qb_s=pa_t+qb_t$. Two such ties  impose two linearly independent equations on $p,q$ forcing them to be both zero. Therefore the there are at most two distinct permutations $\pi_1,\pi_2$ such that $\cs^*(p\ba+q\bb)=p\cs(\ba,\pi_1)+q\cs(\ba,\pi_1)=p\cs(\ba,\pi_2)+q\cs(\ba,\pi_2)$. Moreover $s,t$ should be next to each other in $\pi_1,\pi_2$ and their order is switched in $\pi_1,\pi_2$. Therefore they differ by an adjacent transposition.
	\end{enumerate}
\end{proof}

\begin{remark}
	\label{rmk:one_or_two_permutations}
	From  Proposition~\ref{prop:dual_fixedpq} and Lemma~\ref{lem:atmost_twoties}, the solution $\pi^*$ for the primal program~(\ref{eqn:local_primal}) is either a single permutation or a convex combination of two permutations which differ only by an adjacent transposition (i.e. swapping two elements next to each other).
\end{remark}
 
\begin{remark}
\label{prop:gradient_descent_OPT}
$\Psi(p,q)=\cs^*_\bw(p\ba+q\bb)+f^*(-p,-q)$ is a convex function. The gradient\footnote{or a subgradient at points where $\Psi$ is not differentiable} of $\Psi$ can be calculated efficiently and therefore $\OPT=\min_{p,q\ge  0} \Psi(p,q)$ can be found efficiently using gradient (or subgradient) descent~\cite{Bubeck15}.
\end{remark}
It turns out that there is a much more efficient algorithm to find $\OPT$ using binary search. We  need the notion of a subgradient. For a convex function $g:\R^d \to \R$, the subgradient of $g$ at a point $x\in \R^d$ is defined as $\partial g(x)=\{v: g(x+y)\ge g(x)+\inpro{v}{y} \forall y\}$. It is always a convex subset of $\R^d$. If $g$ is differentiable at $x$ then, $\partial g(x)=\{\grad g(x)\}$.

\begin{proposition}[Binary search to find $\OPT$]
\label{prop:binary_search_OPT}
Suppose $a_1,\dots,a_n$ and $b_1,\dots,b_n$ are integers bounded by $B$ in absolute value. We can solve the primal program (\ref{eqn:local_primal}), the dual program (\ref{eqn:local_dual}) and find $\OPT=\min_{p,q>0} \cs^*_\bw(p\ba+q\bb)+f^*(-p,-q)$ in $O(n\log n\log B)$ time. Moreover there is a strongly polynomial randomized algorithm which runs in $O(n\log^2 n)$ time.\footnote{Strongly polynomial refers to the fact that the running time is independent of $B$ or the actual numbers in $\ba,\bb$. In this model, it is assumed that arithmetic and comparison operations between $a_i$'s and $b_i$'s take constant time.} 
\end{proposition}
\begin{proof}
To solve the primal program (\ref{eqn:local_primal}) and the dual program (\ref{eqn:local_dual}), it is enough to find $(\pi^*,\alpha^*,\beta^*)$ and $(p^*,q^*)$ which together satisfy all the simplified KKT conditions~(\ref{eqn:simplified_KKT}).

Throughout the proof, we  drop the subscript $\bw$ from $\cs_\bw$ for brevity.
	$\Psi(p,q)=\cs(p\ba+q\bb)+f^*(-p,-q)$ is a convex function.  So a local minimum is a global minimum and therefore it is enough to find $(p^*,q^*)$ such that $0\in \partial \Psi(p^*,q^*)$.
	$$0\in \partial \Psi(p^*,q^*) \iff \grad f^*(-p^*,-q^*) \in \partial \cs^*(p^*\ba+q^*\bb).$$ It is easy to see that $\grad f^*(-p,-q)=(\alpha,\beta)\iff \grad f(\alpha,\beta)=(p,q)$. Therefore we  rewrite the optimality condition for $(p^*,q^*)$ as:
	\begin{equation}
	\label{eqn:fixed_point}
		 (p^*,q^*)\in \grad f\left(\partial \cs^*(p^*\ba+q^*\bb)\right).
	\end{equation}
	  Thus we need to find a fixed point for a set-valued map.
	We  begin by calculating the subgradient $\partial \cs^*(p\ba+q\bb)$. 
	Note that 
	$\cs^*(p\ba+q\bb)=p\cs(\ba,\pi)+q\cs(\bb,\pi)$ where $\pi$ is any permutation which sorts $p\ba+q\bb$ in descending order. By Lemma~\ref{lem:atmost_twoties}, there are two cases:
	\begin{enumerate}
		\item[Case 1:]
	 If there are no ties among $p\ba+q\bb$, then the permutation $\pi$ is unique and $\cs^*(p\ba+q\bb)$ is differentiable at $(p,q)$ and $$\partial \cs^*(p\ba+q\bb)=\{\left(\cs(\ba,\pi),\cs(\bb,\pi)\right)\}.$$ 
	 \item[Case 2:] Suppose there are ties among $p\ba+q\bb$. Then there exists exactly two permutations $\pi_1,\pi_2$ such that $\cs^*(p\ba+q\bb)=\cs(p\ba+q\bb,\pi_1)=\cs(p\ba+q\bb,\pi_2)$. In this case, 
$$\partial \cs^*(p\ba+q\bb)=\left\{\mu\left(\cs(\ba,\pi_1),\cs(\bb,\pi_1)\right)+(1-\mu)\left(\cs(\ba,\pi_2),\cs(\bb,\pi_2)\right):\mu\in [0,1]\right\}.$$
	\end{enumerate}

	We  make a few observations. The value of the subgradient $\partial \cs^*(p\ba+q\bb)$ only depends on the ratio of $p$ and $q$, $\lambda=q/p$; this is because the optimal ranking of $p\ba+q\bb$ only depends on the ratio $\lambda$. And as we change this ratio $\lambda$ from $0$ to $\infty$, the subgradient changes at most $\binom{n}{2}$ times. This happens whenever $\lambda$ is such that $a_i+\lambda b_i=a_j+\lambda b_j$ for some $i\ne j$. We call the set $$C=\left\{\frac{a_i-a_j}{b_j-b_i}:1\le i<j\le n,\frac{a_i-a_j}{b_j-b_i}>0\right\},$$ the critical set of $\lambda$'s where the subgradient changes value\footnote{$|C|$ is equal to the number of inversions in $\ba$ w.r.t. $\bb$, also called the Kendall tau distance.}. Let $m=|C|$ and let $\lambda_1<\lambda_2<\dots<\lambda_m$ be an ordering of the critical set $C$, further define $\lambda_0=0$ and $\lambda_{m+1}=\infty$. 
	\begin{figure}[H]
	\caption{The positive quadrant $\Qplus$ is divided into regions $A_i$ and rays $R_i$ based on the values of the subgradient $\partial \cs^*(p\ba+q\bb)$.}
	\label{fig:Regions}
	\centering
	\includegraphics[scale=0.4]{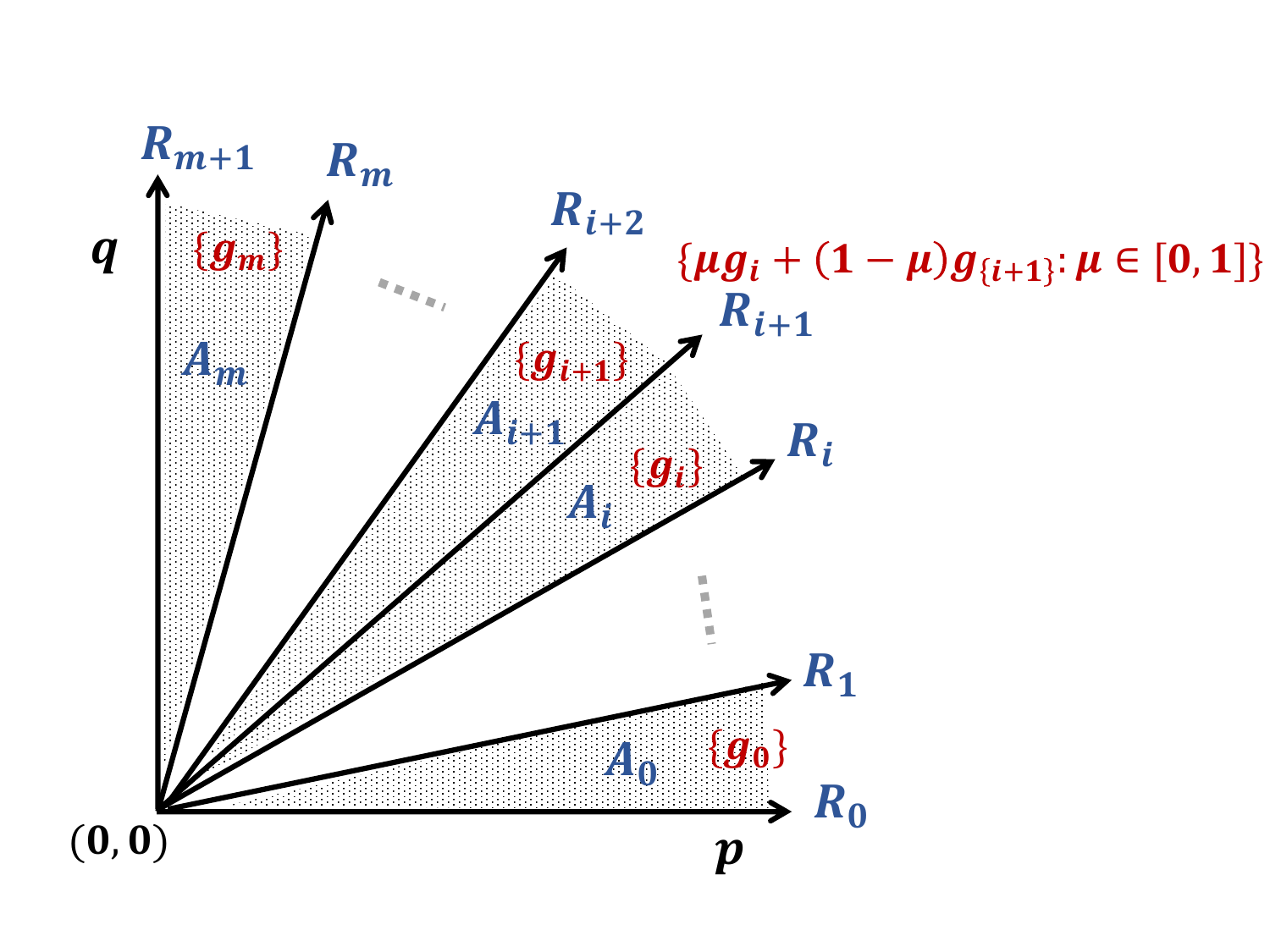}
	\end{figure}
	Define the regions 
	$$A_i=\left\{(p,p\lambda): p>0, \lambda_{i} < \lambda < \lambda_{i+1}\right\}.$$ 
	Also define the rays 
	$$R_i=\left\{(p,p\lambda_i): p>0\right\}.$$ 
	In the region $A_i,$ there is a unique permutation $\sigma_i=\argmax_\pi \cs(p\ba+q\bb,\pi)$,
	therefore the subgradient is unique.
	Denote its value by $g_i$ i.e. 
	$$\partial \cs^*(p\ba+q\bb)|_{A_i}=\{g_i\}=\{(\cs(\ba,\sigma_i),\cs(\bb,\sigma_i))\}.$$ 
	On the ray $R_{i+1},$ we have $\{\sigma_i,\sigma_{i+1}\}=\argmax_\pi \cs(p\ba+q\bb,\pi)$, 
	therefore the subgradient is given by 
	$$\partial \cs^*(p\ba+q\bb)|_{R_{i+1}}=\{\mu g_i+(1-\mu)g_{i+1}: \mu\in[0,1]\}.$$
	Figure~\ref{fig:Regions} shows the values of the subgradient $\partial \cs^*(p\ba+q\bb)$ as a function of $(p,q)$ in the regions $A_i,R_i.$

	Let $\Qplus=\{(p,q):p,q>0\}$ be the positive quadrant. 
	Since $f(\alpha,\beta)$ is strictly increasing in each coordinate in $\Qplus$, 
	the gradients also lie in the positive quadrant, i.e., $\grad f:\Qplus \to \Qplus$.
	 We now define a function $\phi:\{0,1,\dots,m\}\to \{-1,0,1\}$ as follows:
	\begin{equation}
	\label{eqn:phi_binarysearch}
	\begin{aligned}
		\phi(i)=
		\begin{cases}
			0 &\text{ if } \grad f(g_i)\in R_i\cup A_i \cup R_{i+1}\\
			+1 &\text{ if } \grad f(g_i) \text{ lies anticlockwise to } R_i \cup A_i \cup R_{i+1}\\
			-1 &\text{ if } \grad f(g_i) \text{ lies clockwise to } R_i \cup A_i \cup R_{i+1}.
		\end{cases}
	\end{aligned}
	\end{equation}

We  now show that it is enough to find some $i\in \{0,1,\dots,m\}$ such that one of the following is true.
\begin{enumerate}
	\item $\phi(i)=0$. In this case, we  set $(p^*,q^*)=\grad f(g_i).$ The condition $\phi(i)=0$ implies that $(p^*,q^*)\in R_i \cup A_i \cup R_{i+1}.$ Therefore $g_i \in \partial \cs^*(p^*\ba+q^*\bb).$ Applying $\grad f$ on both sides implies that $(p^*,q^*)=\grad f(g_i) \in \grad f(\partial \cs^*(p^*\ba + q^* \bb))$ which is the fixed point condition~(\ref{eqn:fixed_point}). Moreover setting $\pi^*=\sigma_i$ and $(\alpha^*,\beta^*)=g_i$ gives a solution to the simplified KKT conditions~(\ref{eqn:simplified_KKT}). 
	\item $\phi(i)=1, \phi(i+1)=-1$. In this case, we claim that there exists some $\mu^*\in (0,1)$ such that $\grad f(\mu^* g_i +(1-\mu^*) g_{i+1})\in R_{i+1}$. This is because the curve $\gamma:[0,1]\to \Qplus$ given by $\gamma(\mu)=\grad f(\mu g_i +(1-\mu) g_{i+1})$ starts and ends in opposite sides of the ray $R_{i+1}$ as shown in Figure~\ref{fig:Crossing},  so it should cross it at some point $\mu^*\in (0,1)$ which can be found by binary search (here we need continuity of $\grad f$). We  then set $(p^*,q^*)=\grad f (\mu^* g_i +(1-\mu^*) g_{i+1})$. Since $(p^*,q^*)\in R_{i+1}$, $\mu^* g_i +(1-\mu^*) g_{i+1} \in \partial \cs^*(p^*\ba+q^*\bb)$. Applying $\grad f$ to both sides, we get $(p^*,q^*)\in \grad f(\partial \cs^*(p^*\ba+q^*\bb))$ which is the fixed point condition~(\ref{eqn:fixed_point}). Setting $\pi^*=\mu^* \sigma_i + (1-\mu^*)\sigma_{i+1}$ and $(\alpha^*,\beta^*)=\mu^* g_i +(1-\mu^*)g_{i+1}$ gives a solution to the simplified KKT conditions~(\ref{eqn:simplified_KKT}). 
	\begin{figure}[H]
	\caption{The endpoints of the curve $\gamma$ lie on opposite sides of the ray $R_{i+1}$  so $\gamma$ should cross the ray $R_{i+1}$ at some point.}
	\label{fig:Crossing}
	\includegraphics[scale=0.4]{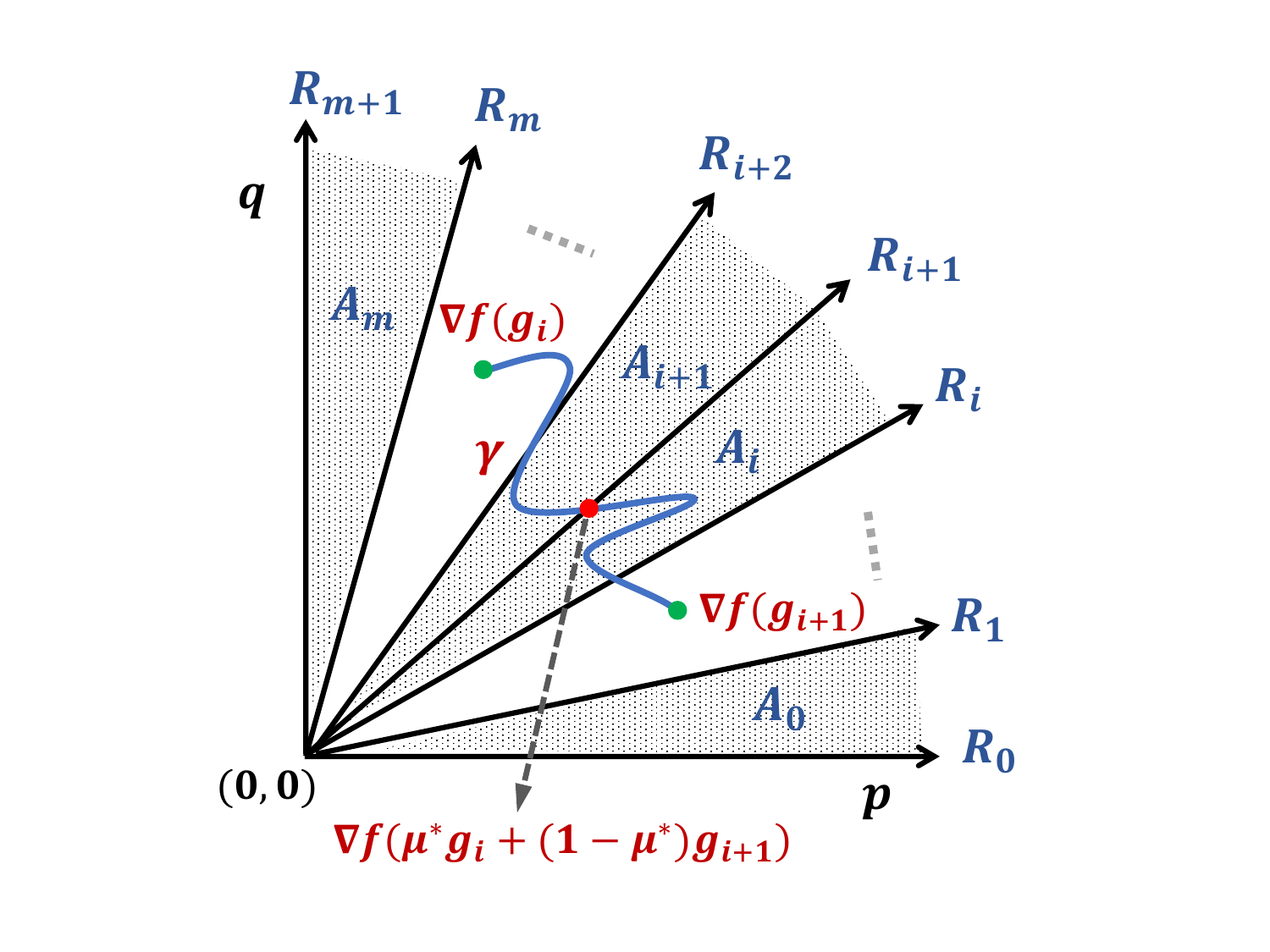}
	\centering
	\end{figure}
\end{enumerate}

Now if either $\phi(0)=0$ or $\phi(m)=0$, we are done. Otherwise, $\phi(0)=1$ and $\phi(m)=-1$. By a simple binary search on $\{0,1,\dots,m\}$ we can find a point $i$ such that either $\phi(i)=0$ or $\phi(i)=1, \phi(i+1)=-1$.

A naive implementation of the above described binary search will require finding the set of critical values $C$ and sorting them. This can take $\Omega(n^2\log n)$ time. To improve this to near-linear time, we need a way do binary search on $C$ without listing all values in $C$. See Appendix~\ref{sec:running_time_binarysearch_local} for how to achieve this.
\end{proof}

We are now ready to prove Theorem~\ref{thm:main_local}.
\begin{proof}[Proof of Theorem~\ref{thm:main_local}]
By Proposition~\ref{prop:binary_search_OPT} and Remark \ref{rmk:one_or_two_permutations}, we  can find a solution $(\pi^*,\alpha^*,\beta^*)$ to the primal program~(\ref{eqn:local_primal}) where $\pi^*$ is either a permutation or a convex combination of two permutations which differ by an adjacent transposition. If $\pi^*$ is a permutation then $$\co^*_\bw(\ba,\bb)\le \OPT=f(\cs_\bw(\ba,\pi^*),\cs_\bw(\bb,\pi^*)).$$ Thus we  just output $\pi^*$ which is the optimal ranking.

If $\pi^*=\mu \pi_1 + (1-\mu) \pi_2$ i.e. a convex combination of $\pi_1,\pi_2$ which differ in the $i,i+1$ positions, then 
\begin{align*}
\co^*_\bw(\ba,\bb)\le \OPT&=f(\mu(\cs_\bw(\ba,\pi_1),\cs_\bw(\bb,\pi_1))+(1-\mu)(\cs_\bw(\ba,\pi_2),\cs_\bw(\bb,\pi_2))\\ &\le f((\cs_{\bw'}(\ba,\pi_1),\cs_{\bw'}(\bb,\pi_1)))=\co_{\bw'}(\ba,\bb,\pi_1)
\end{align*}
 where $\bw'=\bw+(w_i-w_{i+1})\be_{i+1}$. Similarly $\co^*_\bw(\ba,\bb)\le \co^*_{\bw'}(\ba,\bb,\pi_2).$ So in this case we  output either $\pi_1$ or $\pi_2$, whichever has the higher combined objective.

\end{proof}

\section{Multiple instances of ranking with global aggregation}
\label{sec:global}
Suppose we have several instances of $\rank$ and we want to do well locally in each problem, but we also want to do well when we aggregate our solutions globally. Such a situation arises in the example application discussed in the beginning. For each user we want to solve an instance of $\rank$, but globally across all the users, we want that the average revenue and average relevance are high as well. We model this as follows.

Let $(\ba_1,\bb_1,\bw_1,f_1),\dots,(\ba_m,\bb_m,\bw_m,f_m)$ be $m$ instances of $\rank$ on sequences of length $n$.  Let $\barw=(\bw_1,\dots,\bw_m),\ ,\barf=(f_1,\dots,f_m),\ \bara=(\ba_1,\dots,\ba_m),\ \barb=(\bb_1,\dots,\bb_m)$ and let $\barpi=(\pi_1,\dots,\pi_m)$ be a sequence of rankings of $[n]$. The global cumulative a-score and b-score of $\barpi$ is defined as $$\gcs_{\barw}(\bara,\barpi)=\sum_{i=1}^m \cs_{\bw_i}(\ba_i,\pi_i)\text{ and } \gcs_{\barw}(\barb,\pi)=\sum_{i=1}^m \cs_{\bw_i}(\bb_i,\pi_i)$$ respectively. Suppose $F(\alpha,\beta)$ is a concave function increasing in each coordinate. The combined objective function is defined as $$\co_{\barw}(\bara,\barb,\barpi)=F(\gcs_{\barw}(\bara,\barpi),\gcs_{\barw}(\barb,\barpi))+\sum_{i=1}^m f_i(\cs_{\bw_i}(\ba_i,\pi_i),\cs_{\bw_i}(\bb_i,\pi_i)).$$

\begin{definition}[$\multirank(\bara,\barb,\barw,\barf,F)$]
Given $\bara,\barb,\barf,F$, find a sequence of rankings $\barpi=(\pi_1,\dots,\pi_m)$ of $[n]$ which maximizes $\co_{\barw}(\bara,\barb,\barpi)$ i.e. find $\co^*_{\barw}(\bara,\barb)=\max_{\barpi} \co_{\barw}(\bara,\barb,\barpi)$
\end{definition}

We will assume that the functions $f_1,\dots,f_m,F$ are concave and strictly increasing in each coordinate, and that they are differentiable with continuous derivatives. Our main theorem is that we can efficiently find a sequence of rankings $\barpi$ which, with a slight advantage, does as well as the optimal sequence of rankings.

\begin{theorem}
\label{thm:main_global}
Suppose the functions $f_1,\dots,f_m,F$ are concave and strictly increasing in each coordinate.
Given an instance of $\multirank(\bara,\barb,\barw,\barf,F)$, we can efficiently find a sequence of rankings $\barpi=(\pi_1,\dots,\pi_m)$ such that $$\co_{\barw'}(\bara,\barb,\barpi)\ge \co^*_{\barw}(\bara,\barb)$$ where $\barw'=(\bw_1',\dots,\bw_m')$ and $\bw_i'=\bw_i+(w_{t_i}-w_{{t_i}+1})\be_{{t_i}+1}$ for some $t_i\in [n-1]$ i.e., each $\bw'_i$ is obtained by replacing $w_{t_i+1}$ with $w_{t_i}$.
\end{theorem}

In the special case when the weight vectors $\bw_1,\dots,\bw_m$ are just a sequence of $k$ ones followed by zeros i.e. the cumulative scores are calculated by adding the scores of top $k$ results, $\multirank$ is called the $\multiTOP_k$ problem.

\begin{corollary}
\label{cor:main_global_topk}
Suppose the functions $f_1,\dots,f_m,F$ are concave and strictly increasing in each coordinate.
Given an instance of $\multiTOP_k(\bara,\barb,\barf,F)$, we can efficiently find a sequence $\barS=(S_1,\dots,S_m)$ of subsets of $[n]$ of size at most $k+1$ (corresponding to the top $k+1$ elements) such that
$$\co(\bara,\barb,\barS)\ge \max_{|T_i|=k}\co(\bara,\barb,(T_1,\dots,T_m)).$$
\end{corollary}

Our approach to prove Theorem~\ref{thm:main_global} is again very similar to how we proved Theorem~\ref{thm:main_local}. We write a convex programming relaxation and solve its dual program. 
We will also assume that the sequences $\ba_i,\bb_i,\bw$ are generic which can be ensure by perturbing all entries by tiny additive noise, this will not change $\co^*_\bw(\ba,\bb)$ by much. By a limiting argument, this will not affect the result.
We first develop a convex programming relaxation $\OPT$ as shown in~(\ref{eqn:global_primal}).

\begin{equation}
\label{eqn:global_primal}
\begin{aligned}
\OPT=\max_{\pi^i_{jk},\alpha_i,\beta_i,\alpha,\beta\ge 0}\ & F(\alpha,\beta)+\sum_{i=1}^m f_i(\alpha_i,\beta_i)&\\
s.t.\ &\forall i\in[m]\quad  \alpha_i \le \sum_{j,k=1}^n w_{ij}a_{ik}\pi^i_{jk}   &\dash (p_i)\\
&\forall i\in[m]\quad \beta_i \le \sum_{j,k=1}^n w_{ij}b_{ik}\pi^i_{jk}  &\dash (q_i)\\
&\alpha \le \sum_{i=1}^m\sum_{j,k=1}^n w_{ij}a_{ik}\pi^i_{jk}    &\dash (p)\\
&\beta \le \sum_{i=1}^m\sum_{j,k=1}^n w_{ij}b_{ik}\pi^i_{jk}   &\dash (q)\\
&\forall i\in[m], j\in [n]\quad \sum_{k=1}^n \pi^i_{jk} \le 1  &\dash (r_{ij})\\
&\forall i\in[m], k\in [n]\quad \sum_{j=1}^n \pi^i_{jk} \le 1  &\dash (c_{ik})
\end{aligned}
\end{equation}
It is clear that $\OPT$ is a relaxation for $\co^*_\bw(\ba,\bb)$ with $\co^*_\bw(\ba,\bb)\le \OPT.$ By convex programming duality, $\OPT$ can be expressed as a dual minimization problem~(\ref{eqn:global_dual}) by introducing a dual variable for every constraint in the primal as shown in (\ref{eqn:global_primal}). Again by Slater's condition, strong duality holds~\cite{BoydV04}. The constraints in the dual correspond to variables in the primal as shown in~(\ref{eqn:global_dual}).
\begin{equation}
\label{eqn:global_dual}
\begin{aligned}
\OPT=\min_{r_{ij},c_{ik},p_i,q_i,p,q\ge 0}\ & \sum_{i=1}^m\left(\sum_{j=1}^n r_{ij}+\sum_{k=1}^n c_{ik}+f_i^*(-p_i,-q_i)\right)+F^*(-p,-q)&\\
s.t.\ &\forall i\in [m],j,k\in[n]\quad r_{ij}+c_{ik}\ge w_{ij}\left((p_i+p)a_{ik}+(q_i+q)b_{ik}\right) &\dash (x_{ij})
\end{aligned}
\end{equation}

Here $f_i^*$ is the Fenchel dual of $f$ and $F^*$ is the Fenchel dual of $F.$
If $(\pi^*,\alpha_i^*,\beta_i^*,\alpha^*,\beta^*)$ is some optimal solution for the primal~(\ref{eqn:global_primal}) and $(r^*,c^*,p_i^*,q_i^*,p^*,q^*)$ is some optimal solution for the dual~(\ref{eqn:global_dual}), then they should together satisfy the KKT conditions given in~(\ref{eqn:KKT_conditions_global}). Note that a constraint of primal is tight if the corresponding variable in the dual is strictly positive and vice-versa.
\begin{equation}
\label{eqn:KKT_conditions_global}
\begin{array}{ccl|ccl}
p_i^*>0 &\Rightarrow& \sum_{j,k=1}^n w_{ij}a_{ik}\pi^{i*}_{jk}=\alpha_i^*  & \nabla f_i(\alpha_i^*,\beta_i^*)&=&(p_i^*,q_i^*)\\
q_i^*>0 &\Rightarrow& \sum_{j,k=1}^n w_{ij}b_{ik}\pi^{i*}_{jk}=\beta_i^*  & \nabla F(\alpha^*,\beta^*)&=&(p^*,q^*)\\
r^{*}_{ij}>0 &\Rightarrow& \sum_{k=1}^n \pi^{i*}_{jk} = 1 & \pi^{i*}_{jk}>0 &\Rightarrow& r_{ij}+c_{ik}\ge w_{ij}\left((p_i+p)a_{ik}+(q_i+q)b_{ik}\right)\\
c^{*}_{ik}>0 &\Rightarrow& \sum_{j=1}^n \pi^{i*}_{jk} = 1 & & & \\
p^*>0 &\Rightarrow& \sum_{i=1}^m\sum_{j,k=1}^n w_{ij}a_{ik}\pi^{i*}_{jk} = \alpha^* & & & \\
q^*>0 &\Rightarrow& \sum_{i=1}^m\sum_{j,k=1}^n w_{ij}b_{ik}\pi^{i*}_{jk} = \beta^* & & & 
\end{array}
\end{equation}

\begin{proposition}
\label{prop:dual_fixedpq_global}
Let $p,q,p_1,q_1,\dots,p_m,q_m> 0$ be fixed. Then the value of the minimization program in~(\ref{eqn:global_dual}) is given by $$\Psi(p,q,p_1,q_1,\dots,p_m,q_m)=F^*(-p,-q)+\sum_{i=1}^m \cs^*_{\bw_i}\left((p+p_i)\ba_i+(q+q_i)\bb_i\right)+f_i^*(-p_i,-q_i).$$ Moreover the KKT conditions~\ref{eqn:KKT_conditions_global} can be simplified to:
\begin{equation}
	\label{eqn:simplified_KKT_global}
	\begin{aligned}
	&\pi_i^*\in \convexhull\{\pi: \cs^*_\bw((p^*+p_i^*)\ba_i+(q^*+q_i^*)\bb_i)=\cs_\bw((p^*+p_i^*)\ba_i+(q^*+q_i^*)\bb_i,\pi)\},\\
	&\sum_{j,k=1}^n w_{ij}a_{ik}\pi^{i*}_{jk}=\alpha_i^*,\quad \sum_{j,k=1}^n w_{ij}b_{ik}\pi^{i*}_{jk}=\beta_i^*,\\
	&\sum_{i=1}^m\sum_{j,k=1}^n w_{ij}a_{ik}\pi^{i*}_{jk} = \alpha^*,\quad \sum_{i=1}^m\sum_{j,k=1}^n w_{ij}b_{ik}\pi^{i*}_{jk} = \beta^*,\\
	&\grad f_i(\alpha_i^*,\beta_i^*)=(p_i^*,q_i^*),\quad \grad F(\alpha^*,\beta^*)=(p^*,q^*). 
\end{aligned}
\end{equation}
\end{proposition}
\begin{proof}
 The proof is very similar to the proof of Proposition~\ref{prop:dual_fixedpq} where we write a linear program for each sub-problem and the corresponding dual linear program. We will skip the details.
 \end{proof}
 
 \begin{remark}
 \label{rmk:gradient_descent}
$$\Psi(p,q,p_1,q_1,\dots,p_m,q_m)=F^*(-p,-q)+\sum_{i=1}^m \cs^*_{\bw_i}\left((p+p_i)\ba_i+(q+q_i)\bb_i\right)+f_i^*(-p_i,-q_i)$$ is a convex function. The gradient (or a subgradient) of $\Psi$ can be calculated efficiently and therefore $\OPT=\min_{p,q> 0} \Psi(p,q)$ can be found efficiently using gradient (or subgradient) descent. The objective is also amenable to the use of stochastic gradient descent which can be much faster when $m\gg 1.$
 \end{remark}

\begin{proposition}
\label{prop:binary_search_OPT_global}
	For a fixed $p,q>0,$ we can find $$\min_{p_i,q_i> 0} \Psi(p,q,p_1,q_1,\dots,p_m,q_m)=\min_{p_i,q_i> 0} F^*(-p,-q)+\sum_{i=1}^m \cs^*_{\bw_i}\left((p+p_i)\ba_i+(q+q_i)\bb_i\right)+f_i^*(-p_i,-q_i)$$ efficiently using binary search.
\end{proposition}
\begin{proof}
	It is enough to the find $$\argmin_{p_i,q_i>0} \cs^*_{\bw_i}\left((p+p_i)\ba_i+(q+q_i)\bb_i\right)+f_i^*(-p_i,-q_i)$$ for a fixed $i.$ 
	By convexity of the objective, it is enough to find $(p_i,q_i)$ such that $\grad f_i^*(-p_i,-q_i)\in \partial cs^*_{\bw_i}\left((p+p_i)\ba_i+(q+q_i)\bb_i\right)$ which can be rewritten as:
	\begin{equation}
	\label{eqn:fixedpoint_global}
	(p_i,q_i)\in \grad f_i\left(\partial cs^*_{\bw_i}\left((p+p_i)\ba_i+(q+q_i)\bb_i\right)\right).
	\end{equation}
	This fixed point equation can be solved using binary search in exactly the same way as in Proposition~\ref{prop:binary_search_OPT}. $\partial cs^*_{\bw_i}\left((p+p_i)\ba_i+(q+q_i)\bb_i\right)$ only depends on the ration $\lambda=(q+q_i)/(p+p_i)$. Geometrically, this ratio is constant on any line passing through $(-p,-q)$. Figure~\ref{fig:multiranking_binarysearch} shows the regions where $\partial cs^*_{\bw_i}\left((p+p_i)\ba_i+(q+q_i)\bb_i\right)$ remains constant. Now the fixed point equation~(\ref{eqn:fixedpoint_global}) can be solved in exactly the same way as in the proof of Proposition~\ref{prop:binary_search_OPT}. We will skip the details.
		\begin{figure}[H]
	\caption{The positive quadrant $\Qplus$ is divided into regions $A_j$ and rays $R_j$ based on the values of the subgradients $\partial cs^*_{\bw_i}\left((p+p_i)\ba_i+(q+q_i)\bb_i\right)$. One can use binary search as in the proof of Proposition~\ref{prop:binary_search_OPT} to solve the fixed point equation~(\ref{eqn:fixedpoint_global}).}
	\label{fig:multiranking_binarysearch}
	\centering
	\includegraphics[scale=0.4]{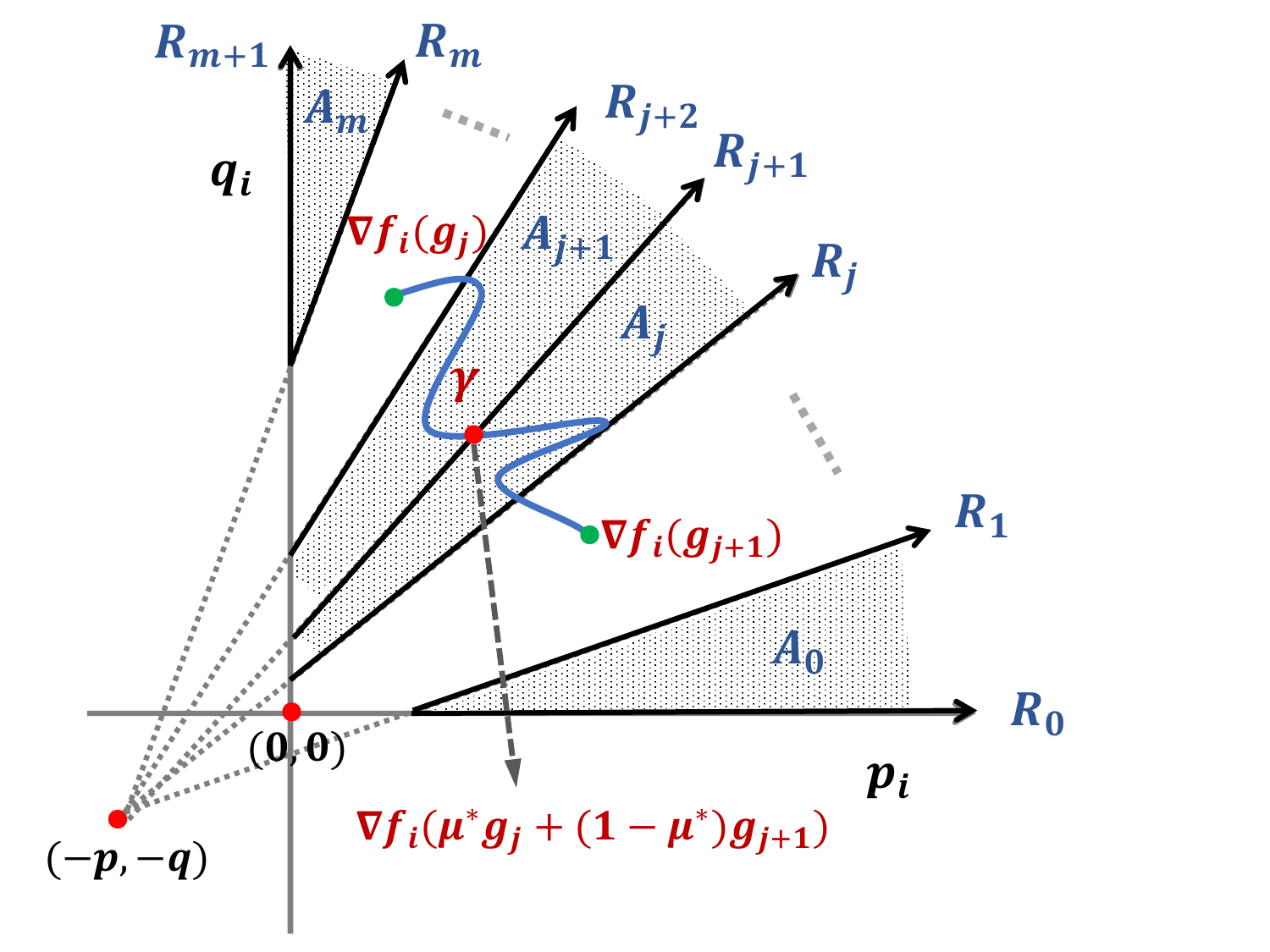}
	\end{figure}
\end{proof}

 We are now ready to prove Theorem~\ref{thm:main_global}.
 \begin{proof}[Proof of Theorem~\ref{thm:main_global}]
By Proposition~\ref{prop:binary_search_OPT_global} and \ref{rmk:one_or_two_permutations}, we can find a solution to the primal program~\ref{eqn:global_primal} where each $\pi_i^*$ is either a permutation or a convex combination of two permutations which differ by an adjacent transposition. If $\pi_i^*$ is a permutation then we just output $\pi^*$ which is the optimal ranking for the $i^{th}$ ranking problem.
If $\pi^*=\mu \pi_1 + (1-\mu) \pi_2$ i.e. a convex combination of $\pi_1,\pi_2$ which differ in the $j,j+1$ positions, then we output either $\pi_1$ or $\pi_2$ as the ranking for the $i^{th}$ subproblem.
 \end{proof}

\section{Experiments}
\label{sec:experiments}

\subsection{Synthetic Data}

We first present the results on synthetic data. 
The purpose of this experiment is to illustrate 
how different objective functions affect the distribution of the NDCGs. 
These results are summarized in Figures~\ref{fig:lognormal} and~\ref{fig:lognormal-cdf}. 
We present the scatter plot of the NDCGs, just like in Figure~\ref{fig:scatter}, as well as the cumulative distribution functions (CDFs). 

We aim to capture the multiple intents scenario where 
the results are likely to be good along one dimension but not both. 
The values are drawn from a log normal anti-correlated distribution: 
$\log(\aij)$ and the $\log(\bij)$s are drawn 
from a multivariate Gaussian with mean zero and covariance matrix 
\[ 
\begin{bmatrix}
 0.2 & -0.16           \\[0.3em]
-0.16  & 0.2          \\[0.3em]
\end{bmatrix}.
\]
A scatter plot of the distribution of these values is shown in Figure~\ref{fig:lognormal-ab-dist}.
\begin{figure}[H]
	\centering
	
	\includegraphics[width=.5\textwidth]{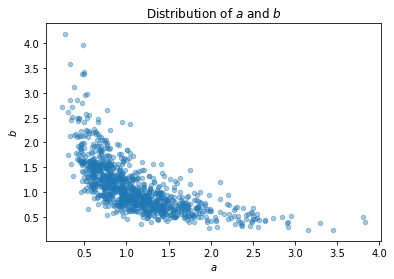}
	\caption{Distribution of $\aij$ and $\bij$ values drawn from an anti-correlated log normal distribution.}
	\label{fig:lognormal-ab-dist} 
\end{figure}
Other parameters of the experiment are as follows: 
we draw 50 results for each instance, i.e., $n = 50$. 
The weight vector is the same as in the Introduction, \eqref{eqn:weight}, 
except that we only consider the top 10 results, i.e., the coordinates of $\bw$ after the first 10 are 0. 
The number of different instances, $m$ is 500.

In addition to the product and the sum, we present the result of 
using two more combining functions: a quadratic and a normalized sum. 
Since we are plotting the NDCGs, a natural algorithm is to maximize the 
sum of the NDCGs. This is what the normalized sum does. 
The quadratic function first normalizes the scores to get the NDCGs, 
and then applies the function 
\[ 
f(x,y) = 2x - x^2 + 2y - y^2. 
\]

It can be seen from Figure \ref{fig:lognormal} 
that the 
concave functions are quite a bit more clustered than the additive functions. 
This can also be seen in the table inside the figure, which shows 
the sum of the cumulative scores, the DCGs, as well as the mean of the 
normalized cumulative scores, the NDCGs. 
These quantities are almost the same across all algorithms. 
We also show the standard deviations of the NDCGs, 
which quite well captures how clustered the points are, 
and shows a significant difference between the concave and the additive functions.

\begin{figure}[H]
	\centering
	
	\includegraphics[width=\textwidth]{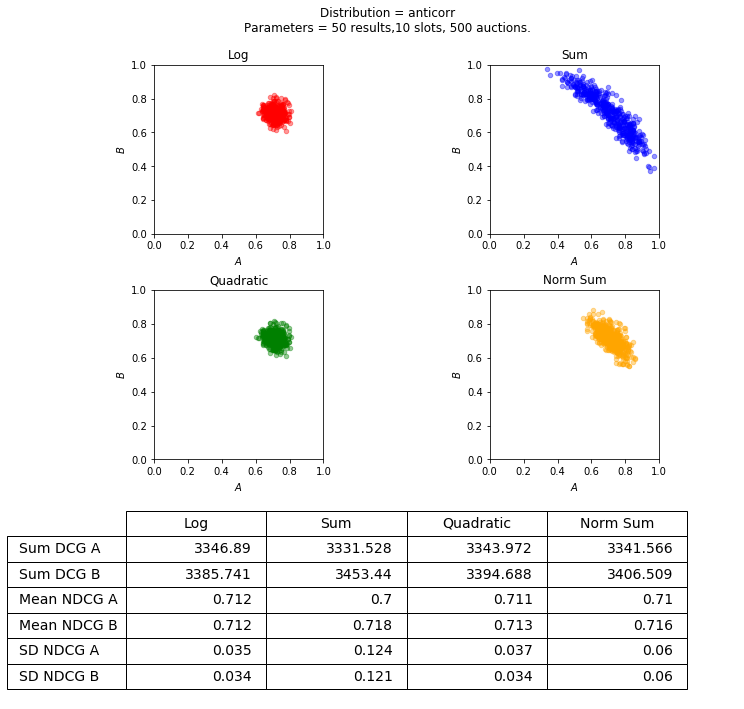}
	\caption{Scatter plot of NDCGs for two different objectives, $A$ and $B$, on synthetic data.}
	\label{fig:lognormal} 
\end{figure}

We present the CDFs of the NDCGs  for the four algorithms in Figure \ref{fig:lognormal-cdf}. 
The dots on the curves represent deciles, i.e., 
the values corresponding to the bottom $10\%$ of the population, 
$20\%$, and so on. 
Recall that a lower CDF implies that the values are higher.\footnote{
	Recall that a distribution $F$ stochastically dominates another distribution $G$ iff the cdf of $F$ is always below that of $G$.}
The CDF shows that in the bottom half of the distribution, 
the concave functions are higher than the additive ones. 
Also the steeper shape of the CDFs for the concave functions
show how they are more concentrated. 
There is indeed a price to pay in that the top half are worse
but this is unavoidable. 
The additive function picks a point on the Pareto frontier after all; 
in fact, it maximizes the mean of $A$ for a fixed mean of $B$ and vice versa. 
The whole point is that the mean is not necessarily an appropriate metric.

\begin{figure}[H]
	\centering
	
	\includegraphics[width=\textwidth]{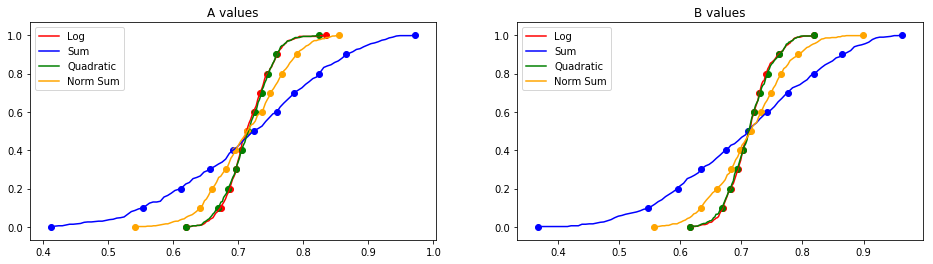}
	\caption{CDFs of NDCGs for two different objectives, $A$ and $B$, on synthetic data.}
	\label{fig:lognormal-cdf} 
\end{figure}

\subsection{Real Data}
The purpose of the experiment in this section is 
to show how the ideas from this paper could help in a realistic setting of ranking. 
We present experiments on real data from LinkedIn, 
sampled from one day of results from their \emph{news feed}.
The number of instances is about 33,000. 
The results are either \emph{organic} or \emph{ads}. 
Organic results only have relevance scores, their revenue scores are 0. 
Ads have both relevance and revenue scores.

The objectives are ad revenue and  relevance.
We will denote the ad revenue by $A$ and relevance by $B$. 
We use the same weight vector $\bw$ in the introduction, \eqref{eqn:weight}, up to 10 coordinates. 
Ad revenue can be added across instances, so we just 
sum up the ad revenue across different instances and 
tune the algorithms so that they all have roughly the same revenue. 
(The difference is less than $1\%$.)
It makes less sense to add the relevance scores. 
In fact it is more important to make sure that no instance gets a really bad relevant score, rather than optimize for the mean or even 
the standard deviation. 
To do this, we aim to make the bottom quartile (25\%) as high as possible.

Motivated by the above consideration, for the relevance objective, we consider a function that has a steep penalty for lower values. 
We first normalize the scores to get the NDCG, and then apply this function on the normalized value. 
For revenue, we just add up the cumulative scores. 
The function we use to combine the two cumulative scores is thus
$$f(x,y) = x - e^{-c_1y/\cs_{\bw}^*(\bb) - c_2} , $$
for some suitable constants $c_1$ and $c_2$. 
Higher values of $c_1$ make the curve steeper and make the distribution more concentrated.  
We choose $c_1$ so that we benefit the bottom quartile as much as possible. 
The constant $c_2$ is tuned so that the total revenue is close to some target. 

We compare this with an additive function. The revenue term is 
not normalized whereas the relevance term is.
This function is 
$$g(x,y) = x + c_3y/\cs_{\bw}^*(\bb)  , $$
for some suitable choice of $c_3$. 
Once again $c_3$ is tuned to achieve a revenue target. 

We also add constraints on the ranking to better reflect the real scenario
(although not exactly the same constraints as used in reality, for confidentiality reasons). 
An ad cannot occupy the first position, and the total number of ads in the top 10 positions is at most 4. 
It is quite easy to see that we can optimize a single objective given these constraints.\footnote{Although our guarantees don't extend, our algorithm extends to handle such constraints, as long as we can solve the problem of optimizing a single objective. In experiments, the algorithm seems to do well. It is an interesting open problem to generalize our guarantees to such settings.} 
We first sort by the score, then slide ads down if the first slot has an ad, and finally remove any ads beyond the top 4. 

We present the CDFs of the NDCGs for relevance for the two algorithms in Figure \ref{fig:cdf}. 
The figure shows that in the bottom quartile the exp function does better, 
and the relation flips above this. 
For the bottom decile, the difference is significant. 
As mentioned earlier, this is exactly what we wanted to achieve.


Another important aspect of a ranking algorithm in this context 
is the set of positions that ads occupy. 
In Figure \ref{fig:addistribution}, we show this distribution: 
for each position, we show the number of instances for which there was an ad in that position. 
For the additive function, which is the graph at the bottom, 
most of the ads are clustered around positions 2 to 4, 
and the number gradually decreases further down. 
The distribution in case of the exp function is better spread out. 
Interestingly, the most common position an ad is shown is the very last one. 
\begin{figure}[H]
	\centering
	
	\includegraphics[width=.7\textwidth]{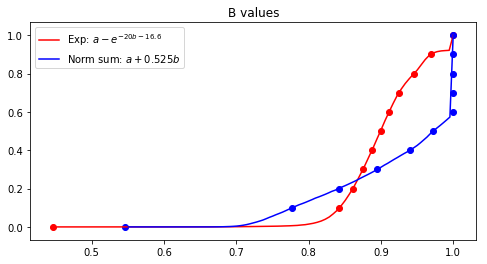}
	\caption{CDF of NDCGs for objectives $B$, relevance, on real data.}
	\label{fig:cdf} 
\end{figure}
\begin{figure}[H]
	\centering
	
	\includegraphics[width=.6\textwidth]{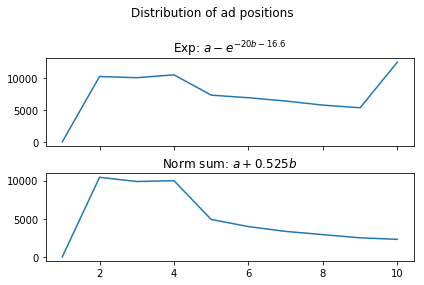}
	\caption{Distribution of ads by position.}
	\label{fig:addistribution} 
\end{figure}

To conclude, the \emph{choice of a concave} function to combine the different objectives gives a greater degree of freedom to ranking  algorithms. 
This freedom can be used to better control several important metrics 
in search and ad rankings. 
This experiment shows how this can be done for the relevance NDCGs in the bottom quartile, or 
for the distribution of ad positions.

\bibliographystyle{plainnat}
\bibliography{references}
\appendix
\clearpage
\section{Running time of binary search in Proposition~\ref{prop:binary_search_OPT}}
 \label{sec:running_time_binarysearch_local}
 \paragraph{Getting $O(n\log n\log B)$ running time.} 
 Note that the critical set $C$ can be of size $\binom{n}{2}$ and doing binary search on this set naively will need us to list all the critical values and then sort them. We can avoid this by doing a binary search directly on the ratio $\lambda=q/p$.

Recall that the critical set $C={\lambda_1,\dots,\lambda_m}$ where $\lambda_1<\lambda_2<\dots<\lambda_m$ and we set $\lambda_0=0,\lambda_{m+1}=\infty$. By the assumption that $a_i,b_i$ are integers bounded by $B$, $\lambda_m < 2B+1$ and $\lambda_{i+1}-\lambda_i\ge 1/8B^2$ for every $i$.
Define $\tilde \phi:\R^{\ge 0}\to \{-1,0,1\}$ as follows:
	\begin{align*}
		\tilde \phi(\lambda)=
		\begin{cases}
			0 &\text{ if } \lambda_i \le \lambda <\lambda_{i+1} \text{ and } \grad f(g_i)\in R_i\cup A_i \cup R_{i+1}\\
			+1 &\text{ if } \lambda_i \le \lambda <\lambda_{i+1} \text{ and }\grad f(g_i) \text{ lies anticlockwise to } R_i \cup A_i \cup R_{i+1}\\
			-1 &\text{ if } \lambda_i \le \lambda <\lambda_{i+1} \text{ and }\grad f(g_i) \text{ lies clockwise to } R_i \cup A_i \cup R_{i+1}.
		\end{cases}
	\end{align*}

\begin{claim}
\label{clm:compute_tildephi}
Given $\lambda^*>0$, we can compute $\tilde\phi(\lambda^*)$ in $O(n\log n)$ time.
\end{claim}
\begin{proof}
It is enough to show that we can find $\lambda_i,\lambda_{i+1}$ such that $\lambda_i \le \lambda^* <\lambda_{i+1}$ in $O(n\log n)$ time. 
We can find the ranking $\pi$ which sorts $\ba+\lambda^* \bb$ in decreasing order in $O(n\log n)$ time. We can then evaluate $g_i=(\cs(\ba,\pi),\cs(\bb,\pi))$ in $O(n)$ time. Now $\lambda_{i+1}$ is the first critical point where the sorted order of $\ba+\lambda\bb$ switches from $\pi$ if we imagine increasing $\lambda$ from $\lambda^*$. Once $\lambda$ crosses $\lambda_{i+1}$, some adjacent elements of $\pi$ switch positions in the sorted order of $\ba+\lambda\bb$. Therefore $$\lambda_{i+1}=\min\left\{\lambda: \lambda>\lambda^* \text{ and } a_{\pi(i)}+\lambda b_{\pi(i)}=a_{\pi(i+1)}+\lambda b_{\pi(i+1)} \text{ for some } i\in[n-1]\right\}$$ where if the set is empty we set $\lambda_{i+1}=\infty$. Note that this can be computed in $O(n)$ time. Similarly $$\lambda_{i}=\min\left\{\lambda: \lambda\le \lambda^* \text{ and } a_{\pi(i)}+\lambda b_{\pi(i)}=a_{\pi(i+1)}+\lambda b_{\pi(i+1)} \text{ for some } i\in[n-1]\right\}$$ where if the set is empty we set $\lambda_{i}=0$. 
\end{proof}

Now we claim that it is enough to find some $\lambda_\ell<\lambda_u$ such that one of the following is true:
\begin{enumerate}
	\item $\tilde\phi(\lambda_\ell)=0$ or $\tilde\phi(\lambda_u)=0$. This is similar to the case when $\phi(i)=0$.
	\item $\tilde\phi(\lambda_\ell)=+1,\tilde\phi(\lambda_u)=-1$ and $\lambda_u-\lambda_\ell<1/8B^2$. There can be at most one critical point between $\lambda_\ell,\lambda_u$ i.e. there exists a unique $i$ such that $\lambda_\ell<\lambda_{i+1}<\lambda_u$. Therefore $\lambda_\ell$ and $\lambda_u$ must belong to adjacent regions. This is similar to the case when $\phi(i)=1,\phi(i+1)=-1.$
\end{enumerate}
Now if either $\tilde\phi(0)=0$ or $\tilde\phi(2B+1)=0$, we are done. Otherwise $\tilde\phi(0)=+1$ and $\tilde\phi(2B+1)=-1$. Using binary search in the range $[0,2B+1]$, one can find such $\lambda_\ell,\lambda_u$ in $O(\log B)$ iterations. Since each iteration runs in $O(n\log n)$ time, the total running time is $O(n\log n \log B).$

\paragraph{Getting strongly polynomial randomized $O(n\log^2 n)$ running time.} We will only give a proof sketch. In this case we cannot do a binary search over $\lambda.$ Because all the critical $\lambda$ can be concentrated in a small region and we may take a long time to find this region. Before we proceed we make a few claims.
\begin{claim}
	\label{clm:inversions}
	Given two (generic) sequences of numbers $\bc$ and $\bd$ of length $n$, let $I$ be the set of inversions of $\bd$ w.r.t $\bc$ i.e. $I=\left\{(i,j):i<j, \frac{c_i-c_j}{d_j-d_i}>0\right\}$. We can find the size $|I|$ in $O(n\log n)$ time and we can sample uniformly at random from $I$ in $O(n\log n)$ time.\footnote{Note that there can be as many as $\binom{n}{2}$ inversions and so we cannot list them all.}
\end{claim}
\begin{proof}
	We only give a proof sketch. Wlog, we can assume that $\bc$ is already sorted by applying the same permutation to both $\bc,\bd.$ We now sort $\bd$ using the merge sort algorithm and it is not hard to see that we can count and sample from inversions during this process.
\end{proof}

\begin{claim}
	\label{clm:random_lambda}
	Given any $\lambda_\ell<\lambda_u$, we can sample uniformly at random from $C\cap [\lambda_\ell,\lambda_u]$ in $O(n\log n)$ time.
\end{claim}
\begin{proof}
	Let $I$ be the set of inversions of $\ba+\lambda_u\bb$ w.r.t $\ba+\lambda_\ell \bb$. We claim that $$C\cap [\lambda_\ell,\lambda_u]=\{\lambda: a_i+\lambda b_i=a_j +\lambda b_j, (i,j)\in I\}.$$ This is because when you imagine increasing $\lambda$ from $\lambda_\ell$ to $\lambda_u$, $C\cap [\lambda_\ell,\lambda_u]$ is the set of critical points where a switch happens in the sorted order of $\ba+\lambda \bb$. Therefore the critical points in $[\lambda_\ell,\lambda_u]$ correspond exactly to the inversions $I.$ For each inversion $(i,j)\in I$, $\frac{a_i-a_j}{b_j-b_i}\in C\cap [\lambda_\ell,\lambda_u]$.
\end{proof}

Suppose $\phi(0)=1$ and $\phi(m)=-1$ (otherwise we are done), where $\phi$ is the function defined in Equation~(\ref{eqn:phi_binarysearch}). We want to find an $i$ such that $\phi(i)=0$ or $\phi(i)=1,\phi(i+1)=-1$. Set $\lambda_\ell=0,\lambda_u=\infty$. From Claim~\ref{clm:random_lambda}, we can sample a uniformly random $\lambda\in C\cap [\lambda_\ell,\lambda_u]$ in $O(n\log n)$ time. Suppose $\lambda=\lambda_i$, then we can find $\lambda_{i+1}$ and $g_i$ as shown in Claim~\ref{clm:compute_tildephi} in $O(n\log n)$ time. Therefore we can evaluate $\phi(i)$ in $O(n\log n)$ time. Now we continue the binary search based on the value of $\phi(i)$ and update the value of the lower bound $\lambda_\ell=\lambda$ or the upper bound $\lambda_u=\lambda$. In each iteration, the random $\lambda\in C\cap [\lambda_\ell,\lambda_u]$ will eliminate constant fraction of points in $C\cap [\lambda_\ell,\lambda_u]$ i.e. the size of $C\cap [\lambda_\ell,\lambda_u]$ shrinks by a constant factor in expectation. Therefore the algorithm should end in $O(\log n)$ iterations with high probability. In fact, we can stop the sampling process once the size of $C\cap [\lambda_\ell,\lambda_u]$ becomes $O(n)$ and then do a regular binary search on them by listing them all. Since the running time of each iteration is $O(n\log n),$ the total running time is $O(n\log^2 n).$


\end{document}